\documentclass{article}
\usepackage{geometry}
\usepackage{hyperref}
\usepackage{chngcntr}
\usepackage{amsmath,amsthm}
\usepackage{tikz}
\usetikzlibrary{calc,arrows}
\usetikzlibrary{shapes,arrows,positioning,decorations.markings}
\tikzset{>=stealth}
\usepackage{multirow}
\usepackage{subfig}

\newcommand{\fl}{\lambda}
\newcommand{\suflink}{\mathcal{\sigma}}
\newcommand{\ST}{ST}

\newcommand*{\suf}{\mathbin{\tikz [baseline=-0.1ex,->, dotted,thick] \draw (0pt,0.5ex) -- (1.1em,0.5ex);}}%

\newtheorem{theorem}{Theorem}
\newtheorem{lemma}{Lemma}
\newtheorem{definition}{Definition}
\newtheorem{proposition}{Proposition}
\newtheorem{observation}{Observation}

\begin{document}
\title{A Suffix Tree Or Not A Suffix Tree?\footnote{An extended abstract of this paper appeared at IWOCA 2014.}}

\author{
\small Tatiana Starikovskaya, National Research University Higher School of Economics (HSE)\thanks{Partly supported by Dynasty Foundation.} \and
\small Hjalte Wedel Vildh{\o}j, Technical University of Denmark, DTU Compute
}

\date{\empty}

\maketitle

\begin{abstract}
In this paper we study the structure of suffix trees. Given an unlabeled tree $\tau$ on $n$ nodes and suffix links of its internal nodes, we ask the question ``Is $\tau$ a suffix tree?", i.e., is there a string $S$ whose suffix tree has the same topological structure as $\tau$? We place no restrictions on $S$, in particular we do not require that $S$ ends with a unique symbol. This corresponds to considering the more general definition of \emph{implicit} or \emph{extended} suffix trees. Such general suffix trees have many applications and are for example needed to allow efficient updates when suffix trees are built online. We prove that $\tau$ is a suffix tree if and only if it is realized by a string $S$ of length $n-1$, and we give a linear-time algorithm for inferring $S$ when the first letter on each edge is known. This generalizes the work of I~et~al.~[Discrete Appl. Math. 163, 2014].
\end{abstract}

\section{Introduction}
The suffix tree was introduced by Peter Weiner in 1973~\cite{WeinerSTconstruction} and remains one of the most popular and widely used text indexing data structures (see~\cite{40yearsTextIndexing} and references therein). In static applications it is commonly assumed that suffix trees are built only for strings with a unique end symbol (often denoted \$), thus ensuring the useful one-to-one correspondance between leaves and suffixes. In this paper we view such suffix trees as a special case and refer to them as \emph{\$-suffix trees}. Our focus is on suffix trees of \emph{arbitrary strings}, which we simply call \emph{suffix trees} to emphasize that they are more general than \$-suffix trees\footnote{In the literature the standard terminology is \emph{suffix trees} for \$-suffix trees and \emph{extended/implicit suffix trees}~\cite{ExtendedST,Gusfield} for suffix trees of strings not ending with \$.}. Contrary to \$-suffix trees, the suffixes in a suffix tree can end in internal non-branching locations of the tree, called \emph{implicit suffix nodes}.

Suffix trees for arbitrary strings are not only a nice generalization, but are required in many applications. For example in online algorithms that construct the suffix tree of a left-to-right streaming text (e.g., Ukkonen's algorithm~\cite{ukkonen:on-line}), it is necessary to maintain the implicit suffix nodes to allow efficient updates. Despite their essential role, the structure of suffix trees is still not well understood. For instance, it was only recently proved that each internal edge in a suffix tree can contain at most one implicit suffix node~\cite{ImplicitSuffixNodes}.

In this paper we prove some new properties of suffix trees and show how to decide whether suffix trees can have a particular structure. Structural properties of suffix trees are not only of theoretical interest, but are essential for analyzing the complexity and correctness of algorithms using suffix trees.

Given an unlabeled ordered rooted tree $\tau$ and suffix links of its internal nodes, the \emph{suffix tree decision problem} is to decide if there exists a string $S$ such that the suffix tree of $S$ is isomorphic to $\tau$. If such a string exists, we say that \emph{$\tau$ is a suffix tree} and that \emph{$S$ realizes $\tau$}. If $\tau$ can be realized by a string $S$ having a unique end symbol \$, we additionally say that \emph{$\tau$ is a \$-suffix tree}.
See \autoref{fig:threetrees} for examples of a \$-suffix tree, a suffix tree, and a tree which is not a suffix tree. In all figures in this paper leaves are black and internal nodes are white.

\begin{figure}[t]
\definecolor{nodeStroke}{HTML}{000000}\definecolor{nodeFill}{HTML}{FFFFFF}\definecolor{edgeStroke}{HTML}{000000}
\centering
\subfloat[]{
\begin{tikzpicture}[yscale=0.13,xscale=0.24,x=1pt,y=1pt,every node/.style={draw=nodeStroke,line width=0.5pt,circle,inner sep=0,minimum size=4pt}] 
\path[use as bounding box] (1250,-330) rectangle (1600,-710);
\draw {[shift={(1406,-208)}] node[fill=white] (n1) {}
  {[shift={(-112,-120)}] node[fill] (n2) {}}
  {[shift={(0,-120)}] node[fill=white] (n8) {}
    {[shift={(-30,-120)}] node[fill] (n9) {}}
    {[shift={(30,-120)}] node[fill] (n10) {}
      {[shift={(-30,-120)}] node[fill] (n17) {}}
      {[shift={(30,-120)}] node[fill] (n18) {}}
    }
  }
  {[shift={(110,-120)}] node[fill=white] (n11) {}
    {[shift={(-30,-120)}] node[fill] (n12) {}}
    {[shift={(30,-120)}] node[fill] (n16) {}}
  }
}
;
\draw[edgeStroke, line width=0.5pt] (n1) -- (n2) (n1) -- (n8) (n8) -- (n9) (n8) -- (n10) (n10) -- (n17) (n10) -- (n18) (n1) -- (n11) (n11) -- (n12) (n11) -- (n16) ;

\draw[->,thick,dotted, bend left] (n8) edge (n1);
\draw[->,thick,dotted] (n11) edge (n8);
\end{tikzpicture}
}
\qquad\qquad
\subfloat[]{
\begin{tikzpicture}[scale=0.13,x=1pt,y=1pt,every node/.style={draw=nodeStroke,line width=0.5pt,circle,inner sep=0,minimum size=4pt}] 
\path[use as bounding box] (1200,-330) rectangle (1600,-710);
\draw {[shift={(1406,-208)}] node[fill=white] (n1) {}
  {[shift={(-92,-120)}] node[fill=white] (n8) {}
    {[shift={(-46,-120)}] node[fill] (n9) {}}
    {[shift={(46,-120)}] node[fill] (n10) {}
      {[shift={(-46,-120)}] node[fill] (n17) {}}
      {[shift={(46,-120)}] node[fill] (n18) {}}
    }
  }
  {[shift={(92,-120)}] node[fill=white] (n11) {}
    {[shift={(-46,-120)}] node[fill] (n12) {}}
    {[shift={(46,-120)}] node[fill] (n16) {}}
  }
}
;
\draw[edgeStroke, line width=0.5pt] (n1) -- (n8) (n8) -- (n9) (n8) -- (n10) (n10) -- (n17) (n10) -- (n18) (n1) -- (n11) (n11) -- (n12) (n11) -- (n16) ;

\draw[->,thick,dotted, bend left] (n8) edge (n1);
\draw[->,thick,dotted] (n11) edge (n8);
\draw[->,thick,dotted] (n10) edge (n11);
\end{tikzpicture}
}
\qquad\qquad
\subfloat[]{
\begin{tikzpicture}[scale=0.13,x=1pt,y=1pt,every node/.style={draw=nodeStroke,circle,inner sep=0,line width=0.5pt,minimum size=4pt}] 
\path[use as bounding box] (1200,-330) rectangle (1600,-710);
\draw {[shift={(1406,-208)}] node[fill=white] (n1) {}
  {[shift={(-92,-120)}] node[fill=white] (n8) {}
    {[shift={(-46,-120)}] node[fill] (n9) {}}
    {[shift={(46,-120)}] node[fill] (n10) {}}
  }
  {[shift={(92,-120)}] node[fill=white] (n11) {}
    {[shift={(-46,-120)}] node[fill=white] (n12) {}
      {[shift={(-46,-120)}] node[fill=white] (n20) {}
        {[shift={(-46,-120)}] node[fill] (n22) {}}
        {[shift={(46,-120)}] node[fill] (n23) {}}
      }
      {[shift={(46,-120)}] node[fill] (n21) {}}
    }
    {[shift={(46,-120)}] node[fill] (n16) {}}
  }
}
;
\draw[edgeStroke, line width=0.5pt] (n1) -- (n8) (n8) -- (n9) (n8) -- (n10) (n1) -- (n11) (n11) -- (n12) (n12) -- (n20) (n20) -- (n22) (n20) -- (n23) (n12) -- (n21) (n11) -- (n16) ;

\draw[->,thick,dotted, bend left] (n20) edge (n12);
\draw[->,thick,dotted, bend left] (n12) edge (n11);
\draw[->,thick,dotted] (n11) edge (n8);
\draw[->,thick,dotted, bend left] (n8) edge (n1);
\end{tikzpicture}
	\label{fig:not_st}
}
\caption{Three potential suffix trees. (a) is a \$-suffix tree, e.g. for \texttt{ababa\$}. (b) is not a \$-suffix tree, but it is a suffix tree, e.g. for \texttt{abaabab}. (c) is not a suffix tree.\label{fig:threetrees}}
\end{figure}
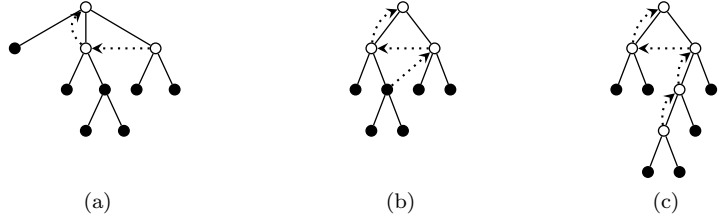

I~et~al.~\cite{InferringStrings} recently considered the suffix tree decision problem and showed how to decide if $\tau$ is a \$-suffix tree in $O(n)$ time, assuming that the first letter on each edge of $\tau$ is also known. Concurrently with our work, another approach was developed in~\cite{Cazaux2014}. There the authors show how to decide if $\tau$ is a \$-suffix tree without knowing the first letter on each edge, but also introduce the assumption that $\tau$ is an unordered tree.

Deciding if $\tau$ is a suffix tree is much more involved than deciding if it is a \$-suffix tree, mainly because we can no longer infer the length of a string that realizes $\tau$ from the number of leaves. Without an upper bound on the length of such a string, it is not even clear how to solve the problem by an exhaustive search. In this paper, we give such an upper bound, show that it is tight, and give a linear time algorithm for deciding whether $\tau$ is a suffix tree when the first letter on each edge is known.

\subsection{Our Results}
In \autoref{sec:suffixtrees}, we start by settling the question of the sufficient length of a string that realizes $\tau$.
\begin{theorem}\label{thm:numberofsuffixnodes}
An unlabeled tree $\tau$ on $n$ nodes is a suffix tree if and only if it is realized by a string of length $n-1$.
\end{theorem}
As far as we are aware, there were no previous upper bounds on the length of a shortest string realizing $\tau$. The bound implies an exhaustive search algorithm for solving the suffix tree decision problem, even when the suffix links are not provided. In terms of $n$, this upper bound is tight, since e.g. stars on $n$ nodes are realized only by strings of length at least $n-1$.

The main part of the paper is devoted to the suffix tree decision problem. We generalize the work of I~et~al.~\cite{InferringStrings} and show in \autoref{sec:algorithm} how to decide if $\tau$ is a suffix tree.
\begin{theorem}\label{thm:main}
Let $\tau$ be a tree with $n$ nodes, annotated with suffix links of internal nodes and the first letter on each edge. There is an $O(n)$ time algorithm for deciding if $\tau$ is a suffix tree.
\end{theorem}
In case $\tau$ is a suffix tree, the algorithm also outputs a string $S$ that realizes $\tau$. To obtain the result, we show several new properties of suffix trees, which may be of independent interest.

\subsection{Related Work} The problem of revealing structural properties and exploiting them to recover a string realizing a data structure has received a lot of attention in the literature. Besides $\$$-suffix trees, the problem has been considered for border arrays~\cite{BorderArray,BorderArrayBoundedAlphabet}, parameterized border arrays~\cite{ParameterizedBorderArraysCounting,ParameterizedBorderArraysVerifying,ParameterizedBorderArrays}, suffix arrays~\cite{GraphsArrays,journals/ita/DuvalL02,SuffixArrays2}, KMP failure tables~\cite{FailureTables1,FailureTables2}, prefix tables~\cite{PrefixTables}, cover arrays~\cite{CoverArray}, directed acyclic word graphs~\cite{GraphsArrays}, and directed acyclic subsequence graphs~\cite{GraphsArrays}.

\section{Suffix Trees\label{sec:suffixtrees}}
In this section we prove \autoref{thm:numberofsuffixnodes} and some new properties of suffix trees, which we will need to prove \autoref{thm:main}. 
We start by briefly recapitulating the most important definitions.

The \emph{suffix tree} of a string $S$ is a compacted trie on suffixes of $S$~\cite{Gusfield}. 
Branching nodes and leaves of the tree are called \emph{explicit nodes}, and positions on edges are called \emph{implicit nodes}. The \emph{label} of a node $v$ is the string on the path from the root to $v$, and the length of this label is called the \emph{string depth} of $v$. The \emph{suffix link} of an internal explicit node $v$ labeled by $a_1 a_2 \ldots a_m$ is a pointer to the node $u$ labeled by $a_2 a_3 \ldots a_m$. We use the notation $v \suf u$ and extend the definition of suffix links to leaves and implicit nodes as well. We will refer to nodes that are labeled by suffixes of $S$ as \emph{suffix nodes}. All leaves of the suffix tree are suffix nodes, and unless $S$ ends with a unique symbol $\$$, some implicit nodes and internal explicit nodes can be suffix nodes as well. 
Suffix links for suffix nodes form a path starting at the leaf labeled by $S$ and ending at the root. Following~\cite{ImplicitSuffixNodes}, we call this path the \emph{suffix chain}.

\begin{lemma}[\cite{ImplicitSuffixNodes}]
\label{lm:suffix_chain}
The suffix chain of the suffix tree can be partitioned into the following consecutive segments: (1) Leaves; (2) Implicit suffix nodes on leaf edges; (3) Implicit suffix nodes on internal edges; and (4) Suffix nodes that coincide with internal explicit nodes. (See~\autoref{fig:suffix_chain_a}.)
\end{lemma}

\begin{figure}
\centering
          \subfloat[]{
          	\label{fig:suffix_chain_a}
 \begin{tikzpicture}[scale=0.8]
 \path[use as bounding box] (-4,0) rectangle (4,-5);
                \draw[every node/.style={draw, circle,minimum size=4pt,inner sep=0pt}] (0,0) node (root) {}
                      (root) +(-160:2cm) node (a) {}
                              (a) +(-110:5cm) node[fill] (aababaababaa) {} 
                              (a) +(-70:1cm) node (aba) {}
                                      (aba) +(-110:3.4cm) node[fill] (abaababaababaa) {} 
                                      (aba) +(-70:3.4cm) node[fill] (ababaababaa) {} 
                      (root) +(-20:2cm) node (ba) {}
                              (ba) +(-110:3.7cm) node[fill] (baababaababaa) {} 
                              (ba) +(-70:3.5cm) node[fill] (babaababaa) {} 
                 ;
                 \draw[every node/.style={draw, circle,fill=black,minimum size=2pt,inner sep=0pt}]
                        ($(a)!0.1!(aababaababaa)$) node (aa) {} 
                        ($(ba)!0.2!(baababaababaa)$) node (baa) {} 
                        ($(ba)!0.5!(babaababaa)$) node (babaa) {} 
                        ($(baa)!0.5!(baababaababaa)$) node (baababaa) {} 
                        ($(aba)!0.2!(abaababaababaa)$) node (abaa) {} 
                        ($(aba)!0.4!(ababaababaa)$) node (ababaa) {} 
                        ($(aa)!0.5!(aababaababaa)$) node (aababaa) {} 
                        ($(abaa)!0.5!(abaababaababaa)$) node (abaababaa) {} 
                 ;
                \path[every node/.style={fill=white,sloped,anchor=center,allow upside down,font=\scriptsize,inner sep=0pt}]
                   (a) edge node {a} (root)
                   (a) edge node {a} (aa)
                   (aa) edge node {babaa} (aababaa)
                   (aababaa) edge node {babaa}  (aababaababaa)
                   (a) edge node {ba} (aba)
                   (aba) edge node {a} (abaa)
                   (abaa) edge node {babaa} (abaababaa)
                   (abaababaa) edge node {babaa} (abaababaababaa)
                   (aba) edge node {baa} (ababaa)
                   (ababaa) edge node {babaa}  (ababaababaa)
                   (root) edge node {ba} (ba)
                   (ba) edge node {a} (baa)
                   (baa) edge node {babaa} (baababaa)
                   (baababaa) edge node {babaa} (baababaababaa)
                   (ba) edge node {baa} (babaa)
                   (babaa) edge node {babaa} (babaababaa)
                  ;
                  %
		  \begin{scope}[opacity=1]
                  	\draw[->,black, dotted, thick] (abaababaababaa)--(baababaababaa);
                   	\draw[->,black, dotted, thick] (baababaababaa)--(aababaababaa);
                   	\draw[->,black, dotted, thick] (aababaababaa)--(ababaababaa);
                   	\draw[->,black, dotted, thick] (ababaababaa)--(babaababaa);                      
                   	\draw[->,black, dotted, thick] (babaababaa)--(abaababaa);
                   	\draw[->,black, dotted, thick] (abaababaa)--(baababaa);
                   	\draw[->,black, dotted, thick] (baababaa)--(aababaa);
                   	\draw[->,black, dotted, thick] (aababaa)--(ababaa);
                   	\draw[->,black, dotted, thick] (ababaa)--(babaa);                   	                   	
                   	\draw[->,black, dotted, thick] (babaa)--(abaa);
                   	\draw[->,black, dotted, thick] (abaa)--(baa);
                   	\draw[->,black, dotted, thick] (baa)--(aa);                    	                   	
                   	\draw[->,black, dotted, thick] (aa)  to[out=200,in=220]  (a);
                   	\draw[->,black, dotted, thick,bend left] (a)  edge  (root); 
		\end{scope}
                \end{tikzpicture}
          }
          \subfloat[]{
          	\label{fig:suffix_chain_b}
 \begin{tikzpicture}[scale=0.8]
 \path[use as bounding box] (-4,0) rectangle (4,-5);
                \draw[every node/.style={draw, circle,minimum size=4pt,inner sep=0pt}] (0,0) node (root) {}
                      (root) +(-160:2cm) node (a) {}
                              (a) +(-110:3.6cm) node[fill] (aabab) {} 
                              (a) +(-70:1cm) node (aba) {}
                                      (aba) +(-110:2cm) node[fill] (abaabab) {} 
                                      (aba) +(-70:2cm) node[fill] (abab) {} 
                      (root) +(-20:2cm) node (ba) {}
                              (ba) +(-110:2cm) node[fill] (baabab) {} 
                              (ba) +(-70:2cm) node[fill] (bab) {} 
                 ;
                 \draw[every node/.style={draw, circle,fill=black,minimum size=2pt,inner sep=0pt}]
                        ($(root)!0.5!(ba)$) node (b) {} 
                        ($(a)!0.5!(aba)$) node (ab) {} 
                        
                 ;
                \path[every node/.style={fill=white,sloped,anchor=center,allow upside down,font=\scriptsize,inner sep=0pt}]
                   (a) edge node {a} (root)
                      (a) edge node {abab} (aabab)
                      (a) edge node {b} (ab)
                      (ab) edge node {a} (aba)
                          (aba) edge node {abab} (abaabab)
                          (aba) edge node {b} (abab)
                   (root) edge node {b} (b) (b) edge node {a} (ba)
                      (ba) edge node {abab} (baabab)
                      (ba) edge node {b} (bab)
                  ;
%
		  \begin{scope}[opacity=1]
                 \draw[->,black, dotted, thick] (abaabab)--(baabab);
                 \draw[->,black, dotted, thick] (baabab)--(aabab);
                 \draw[->,black, dotted, thick] (aabab)--(abab);
                 \draw[->,black, dotted, thick] (abab)--(bab);
                 \draw[->,black, dotted, thick] (bab)--(ab);
                 \draw[->,black, dotted, thick] (ab)--(b);
                 \draw[->,black, dotted, thick] (b)  to[out=100,in=30]  (root);
		 \end{scope}
                \end{tikzpicture}
          }

\caption{(a) The suffix tree $\tau$ of a string $S = abaababaababaa$ with suffix nodes and the suffix chain. (b) The suffix tree of a prefix $S' = abaabab$ of $S$. Suffix links of internal nodes are not shown, but they are the same in both trees.}

\end{figure}

\noindent The string $S$ is fully specified by the order in which the suffix chain visits the subtrees hanging off the root. More precisely,

\begin{observation}\label{obs:suffix_chain}
If $y_0 \suf y_1 \suf \ldots \suf y_l = \mbox{root}$ is the suffix chain in the suffix tree of a string $S$, then $|S|=l$ and $S[i] = a_i$, where $a_i$ is the first letter on the edge going from the root to the subtree containing $y_{i-1}$, $i=1,\ldots,l$.
\end{observation}

We define the parent $par (x)$ of a node $x$ to be the deepest explicit node on the path from the root to $x$ (excluding $x$). The distance between a node and one of its ancestors is defined to be the difference between the string depths of these nodes.

\begin{lemma}\label{lm:pardist}
If $x_1 \suf x_2$ is a suffix link, then the distance from $x_1$ to $par (x_1)$ cannot be less than the distance from $x_2$ to $par (x_2)$.
\end{lemma}
\begin{proof}
If $d$ is the distance between $x_1$ and $par(x_1)$, then the suffix link of $par(x_1)$ points to an explicit ancestor $d$ characters above $x_2$.
\end{proof}

\begin{lemma}
\label{lm:LeafEdges}
Let $x$ be an implicit suffix node. The distance between $x$ and $par (x)$ is not bigger than the length of any leaf edge.
\end{lemma}
\begin{proof}
It follows from \autoref{lm:pardist} that as the suffix chain $y_0 \suf y_1 \suf \ldots \suf y_l = \mbox{root}$ is traversed, the distance from each node to its parent is non-increasing. Since the leaves are visited first, the distance between any implicit suffix node and its parent cannot exceed the length of a leaf edge.
\end{proof}

\begin{lemma}
\label{lm:good_implicit_ST}
If $\tau$ is a suffix tree, then it can be realized by some string such that 
\begin{itemize}
	\item[(1)] The minimal length of a leaf edge of $\tau$ will be equal to one;
	\item[(2)] Any edge of $\tau$ will contain at most one implicit suffix node at the distance one from its upper end.
\end{itemize}
\end{lemma}
\begin{proof}
Let $S$ be a string realizing $\tau$, and $m$ be the minimal length of a leaf edge of $\tau$. Consider a prefix $S'$ of $S$ obtained by deleting its last $(m-1)$ letters. Its suffix tree is exactly $\tau$ trimmed at height~$m-1$. (See~\autoref{fig:suffix_chain_b}.) The minimal length of a leaf edge of this tree is one. Applying~\autoref{lm:LeafEdges}, we obtain that the distance between any implicit suffix node $x$ of this tree and $par (x)$ is one, and, consequently, any edge contains at most one implicit suffix node.
\end{proof}

\begin{lemma}\label{lm:longstrings}
If $\tau$ is realized by a string of length $l$, then it is also realized by strings of length $l+1,l+2,l+3,$ and so on.
\end{lemma}
\begin{proof}
Let $y_0 \suf y_1 \suf \ldots \suf y_{l} = \mbox{root}$ be the suffix chain for a string $S$ that realizes $\tau$. Moreover let $letters(y_i)$ be the set of first letters immediately below node $y_i$. Then $letters(y_{i-1}) \subseteq letters(y_i)$, $i=1,\ldots,l$. Let $y_j$ be the first non-leaf node in the suffix chain (possibly the root). It follows that $Sa$ also realizes~$\tau$, where $a$ is any letter in $letters(y_j)$.
\end{proof}

\noindent We now prove \autoref{thm:numberofsuffixnodes} by showing that if $\tau$ is a suffix tree then a string of length $n-1$ realizes it. By \autoref{lm:good_implicit_ST}, $\tau$ can be realized by a string $S'$ so that the minimal length of a leaf edge is $1$. Consider the last leaf~$\ell$ visited by the suffix chain in the suffix tree of $S'$. By the property of $S'$ the length of the edge $(par(\ell) \rightarrow \ell)$ is 1. Remember that a suffix link of an internal node always points to an internal node and that suffix links cannot form cycles. Moreover, upon transition by a suffix link the string depth decreases exactly by one. Hence if $\tau$ has $I$ internal nodes then the string depth of the parent of $\ell$ is at most $I-1$ and the string depth of $\ell$ is at most $I$. Consequently, if $L$ is the number of leaves in $\tau$, the length of the suffix chain and thus the length of $S'$ is at most $L+I-1 = n-1$, so by \autoref{lm:longstrings} there is a string of this length that realizes $\tau$.

\section{The Suffix Tour Graph}\label{sec:suffix_tour_graph}
In their work~\cite{InferringStrings} I et al. introduced a notion of \emph{suffix tour graphs}. They showed that suffix tour graphs of \$-suffix trees must have a nice structure which ties together the suffix links of the internal explicit nodes, the first letters on edges, and the order of leaves of $\tau$~--- i.e., which leaf corresponds to the longest suffix, which leaf corresponds to the second longest suffix, and so on. Knowing this order and the first letters on edges outgoing from the root, it is easy to infer a string realizing $\tau$. We study the structure of suffix tour graphs of suffix trees. We show a connection between suffix tour graphs of suffix trees and \$-suffix trees and use it to solve the suffix tree decision problem.

Let us first formalize the input to the problem. Consider a tree $\tau = (V,E)$ annotated with a set of suffix links $\suflink: V \rightarrow V$ between internal explicit nodes, and the first letter on each edge, given by a labelling function $\fl : E \rightarrow \Sigma$ for some alphabet $\Sigma$. For ease of description, we will always augment $\tau$ with an auxiliary node $\perp$, the parent of the root. We add the suffix link $(\mbox{root } \suf\perp)$ to $\suflink$ and label the edge $(\perp \rightarrow \mbox{root})$ with a symbol ''?'', which matches any letter of the alphabet. 

To construct the suffix tour graph of $\tau$, we first compute values $\ell(x)$ and $d(x)$ for every explicit node $x$ in $\tau$. The value $\ell(x)$ is equal to the number of leaves $y$ where $par(y) \suf par(x)$ is a suffix link in $\suflink$, and $\fl (par (y) \rightarrow y) = \fl (par (x) \rightarrow x)$. See \autoref{fig:l(x)} for an example. Let $L_x$ and $V_x$ be the sets of leaves and nodes, respectively, of the subtree of $\tau$ rooted at a node~$x$. Note that $L_x$ is a subset of $V_x$. We define $d(x) = |L_x| - \sum_{\substack{y \in V_x}} \ell (y)$. See \autoref{fig:d(x)} for an example.

\begin{figure}
\centering
\subfloat[\label{fig:l(x)}]{
\begin{tikzpicture}[scale=1.5]
\path[use as bounding box] (-0.25,3.3) rectangle (2.75,0.5);
	\path[every node/.style={fill=white,sloped,anchor=center,allow upside down,font=\scriptsize,inner sep=0pt}] 
	(0.8,1) edge node {$\ldots$a} (1,2)
	(1.8,1.3) edge node {$\ldots$a} (2,2.3)
	;
	\draw[fill=black] (0.8,1) circle (0.3ex);
	\node[below] at (0.8,1) {$y$};
	\draw[fill=white] (1,2) circle (0.3ex);
%
	\draw[black,opacity=1, ->, shorten >= 0.6ex, shorten <= 0.6ex,dotted,thick] (1,2) -- (2,2.3); 
%
        \draw (1.8,1.3) -- (1.6,0.8) -- (2,0.8) -- cycle;
	\draw[fill=white] (1.8,1.3) circle (0.3ex);
	\node[right] at (1.8,1.3) {$x$};

	\draw[fill=white] (2,2.3) circle (0.3ex);
%
%
%
\end{tikzpicture}
}\qquad\qquad\qquad
\subfloat[\label{fig:d(x)}]{
\begin{tikzpicture}[scale=1.5]
\path[use as bounding box] (0.5,1) rectangle (3.5,-1.8);
                \draw[every label/.style={draw=none,outer sep=2pt,rectangle,font=\scriptsize}, every node/.style={draw, circle,minimum size=4pt,inner sep=0pt}] (2,0) node[label={[label distance=3pt] right:{$(2,0)$}}] (root) {}
                      (root) +(90:0.7cm) node[label=right:$\perp$] (perp) {}
                      (root) +(-30:1cm) node[fill,label=right:{$(1,0)$}] (b) {}
                      (root) +(-150:1cm) node[fill,label=left:{$(0,1)$}] (leaf) {}
					  (root)+(-90:1cm) node[label=right:{$(1,1)$}] (a) {}node[label={[font=\normalsize]below:{$x$}}] () {}
					  (a)+(-40:1cm) node[fill,label=right:{$(0,1)$}] (ab) {}
					  (a)+(-140:1cm) node[fill,label=left:{$(0,1)$}] (aa) {} node[fill,label={[font=\normalsize] below:{$y$}}] () {}
                 ;
                \path[every node/.style={fill=white,sloped,anchor=center,allow upside down,font=\scriptsize,inner sep=0pt}]
                   (root) edge node[sloped=false, inner sep=1pt] {?} (perp)
                   (a) edge node {$\ldots$a} (root)
                   (leaf) edge node {$\ldots\$$} (root)
                   (root) edge node {b$\ldots$} (b)
                   (a) edge node {b$\ldots$} (ab)
				   (aa) edge node {$\ldots$a} (a)
                  ;
                  \draw[->,thick, dotted, bend left] (a) edge (root);
                  \draw[->,thick, dotted, bend left] (root) edge (perp);
\end{tikzpicture}
}
\caption{(a) An example of a node $x$ with $\ell(x) = 1$. The leaf $y$ contributes to $\ell(x)$ since $\fl (par (y) \rightarrow y) = \fl (par (x) \rightarrow x) = a$. (b) An input consisting of a tree, suffix links and the first letter on each edge. The tree has been extended with the node $\perp$, and each node is assigned values $(\ell(x), d(x))$. For the node $x$, $\ell(x) = 1$, $|L_x| = 2$, and hence $d(x) = 2 - (1 + 0 + 0) = 1$. \label{fig:succpointers}}
\end{figure}
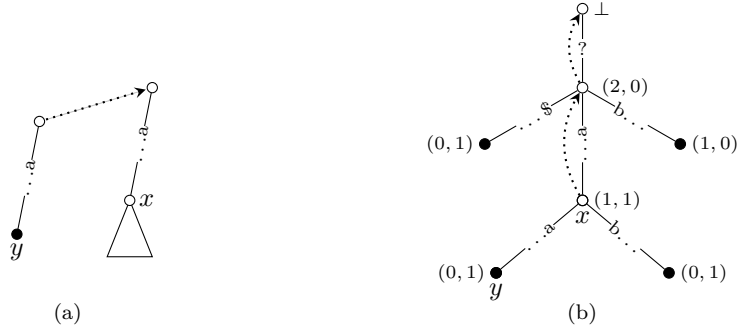

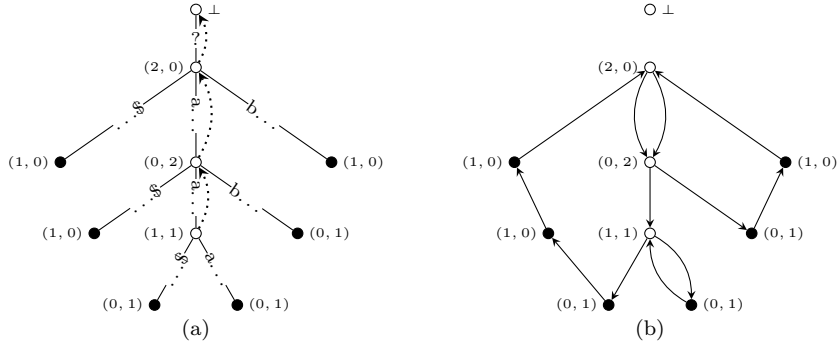
\begin{figure}
\centering
\subfloat[]{
\begin{tikzpicture}[scale=1.1]
\path[use as bounding box] (-2,0.5) rectangle (2,-2.9);
                \draw[every label/.style={draw=none,outer sep=2pt,rectangle,font=\tiny}, every node/.style={draw, circle,minimum size=4pt,inner sep=0pt}] (0,0) node[label=left:{$(2,0)$}] (root) {}
                      (root) +(90:0.7cm) node[label=right:$\perp$] (perp) {}
                      (root) +(-145:2cm) node[fill,label=left:{$(1,0)$}] (0) {} 
                      (root) +(-35:2cm) node[fill,label=right:{$(1,0)$}] (baaa0) {} 
                      ($(0)!0.5!(baaa0)$) node[label=left:{$(0,2)$}] (a) {}
                              (a) +(-145:1.5cm) node[fill,label=left:{$(1,0)$}] (a0) {} 
                              (a) +(-35:1.5cm) node[fill,label=right:{$(0,1)$}] (abaaa0) {} 
                              ($(a0)!0.5!(abaaa0)$) node[label=left:{$(1,1)$}] (aa) {}
                                      (aa) +(-120:1cm) node[fill,label=left:{$(0,1)$}] (aa0) {} 
                                      (aa) +(-60:1cm) node[fill,label=right:{$(0,1)$}] (aaa0) {} 
                 ;
                \path[every node/.style={fill=white,sloped,anchor=center,allow upside down,font=\scriptsize,inner sep=0pt}]
                   (root) edge node[sloped=false, inner sep=1pt] {?} (perp)
                   (0) edge node {$\ldots$\$} (root)
                   (root) edge node {a$\ldots$} (a)
                   (root) edge node {b$\ldots$} (baaa0)
                   (a0) edge node {$\ldots$\$} (a)
                   (a) edge node {a$\ldots$} (aa)
                   (a) edge node {b$\ldots$} (abaaa0)
                   (aa0) edge node {$\ldots$\$} (aa)
                   (aa) edge node {a$\ldots$} (aaa0)
                  ;
                  \draw[->,black,thick,dotted, bend right]
                  (aa) edge (a)
                  (a) edge (root)
                  (root) edge (perp)
                  ;
\end{tikzpicture}
}
\qquad\qquad
\subfloat[]{
     \begin{tikzpicture}[scale=1.1]
 \path[use as bounding box] (-2,0.5) rectangle (2,-2.9);
                \draw[every label/.style={draw=none,outer sep=2pt,rectangle,font=\tiny}, every node/.style={draw, circle,minimum size=4pt,inner sep=0pt}] (0,0) node[label=left:{$(2,0)$}] (root) {}
                      (root) +(90:0.7cm) node[label=right:$\perp$] (perp) {}
                      (root) +(-145:2cm) node[fill,label=left:{$(1,0)$}] (0) {} 
                      (root) +(-35:2cm) node[fill,label=right:{$(1,0)$}] (baaa0) {} 
                      ($(0)!0.5!(baaa0)$) node[label=left:{$(0,2)$}] (a) {}
                              (a) +(-145:1.5cm) node[fill,label=left:{$(1,0)$}] (a0) {} 
                              (a) +(-35:1.5cm) node[fill,label=right:{$(0,1)$}] (abaaa0) {} 
                              ($(a0)!0.5!(abaaa0)$) node[label=left:{$(1,1)$}] (aa) {}
                                      (aa) +(-120:1cm) node[fill,label=left:{$(0,1)$}] (aa0) {} 
                                      (aa) +(-60:1cm) node[fill,label=right:{$(0,1)$}] (aaa0) {} 
                 ;
%
                  \draw[->]
                  (root) edge[bend left] (a)
                  (a) edge (abaaa0)
                  (abaaa0) edge (baaa0)
                  (baaa0) edge (root)
                  (root) edge[bend right] (a)
                  (a) edge (aa)
                  (aa) edge[bend left] (aaa0)
                  (aaa0) edge[bend left] (aa)
                  (aa) edge (aa0)
                  (aa0) edge (a0)
                  (a0) edge (0)
                  (0) edge (root)
                  ;
                \end{tikzpicture}
}

\caption{(a) An input consisting of a tree, suffix links and the first letter on each edge. The input has been extended with the auxiliary node $\perp$, and each node is assigned values $(\ell(x), d(x))$.  (b) The corresponding suffix tour graph. The input (a) is realized by the string \texttt{abaaa\$}, which corresponds to an Euler tour of (b).\label{fig:stgexample}}
\end{figure}

\begin{definition}
\label{def:suffix_tour_graph}
The \emph{suffix tour graph} of a tree $\tau = (V, E)$ is a directed graph $G = (V, E_G)$, where
$E_G =  \{(y \rightarrow x)^k \; | \; (y \rightarrow x) \in E, k = d(x)\} \cup \{(y \rightarrow x) \; | \; y \mbox{ is a leaf contributing to } \ell(x)\}$. Here $(y \rightarrow x)^k$ means the edge $y \rightarrow x$ with multiplicity $k$. If $k = d(x) <0$, we define $(y \rightarrow x)^k$ to be $(x \rightarrow y)^{|k|}$.
\end{definition}
\noindent To provide some intuition of this definition, first recall that the suffix links of the leaves of $\tau$ are not part of this input. In fact the problem of deciding whether $\tau$ is a suffix tree reduces to inferring the suffix links of the leaves of $\tau$, since by knowing these we can reconstruct the suffix chain, and thus also a string realizing $\tau$. The purpose of the suffix tour graph is to encode the constraints that the known suffix links of the internal nodes impose on the unknown suffix links of the leaves as follows: Each leaf $y$ has an outgoing edge that points to the subtree that must contain the leaf immediately after $y$ in the suffix chain. This subtree is uniquely defined by the suffix link of $par(y)$ and the first letter on the edge between $par(y)$ and $y$. For example the outgoing edge for the leaf $y$ in \autoref{fig:l(x)} would point to the subtree rooted in $x$. The value $\ell(x)$ is simply the number of leaves that points to $x$. It can happen that the outgoing edge of $y$ points to another leaf, in which case we then know the successor suffix of $y$ with certainty. The remaining edges in the suffix tour graph are introduced to make the graph Eulerian. The subtree rooted in a node $x$ will have $|L_x|$ outgoing pointers, and $\sum_{y \in V_x} \ell(y)$ incoming pointers, and hence we create $d(x) = |L_x| - \sum_{y \in V_x} \ell(y)$ edges from $par(x)$ to $x$. The main idea, which we will elaborate on in the next section, is that if the graph is Eulerian (and connected), we can reconstruct the suffix chain on the leaves by finding an Eulerian cycle through the leaves of the suffix tour graph. See \autoref{fig:stgexample} for an example of the suffix tour graph.

\begin{lemma}[\cite{InferringStrings}]\label{lm:eulerian}
The suffix tour graph $G$ of a suffix tree $\tau$ is an Eulerian graph (possibly disconnected).
\end{lemma}
\begin{proof}
I et al.~\cite{InferringStrings} only proved the lemma for $\$$-suffix trees, but the proof holds for suffix trees as well. We give the proof here for completeness and because I et al. use different notation. To prove the lemma it suffices to show that for every node the number of incoming edges equals the number of outgoing edges. 

Consider an internal node $x$ of $\tau$. It has $\sum_{\substack{z \in children(x)}} d (z)$ outgoing edges and $\ell(x) + d(x)$ incoming edges. But, $\ell(x) + d(x)$ equals

$$|L_x| - \sum_{\substack{y \in V_x \setminus \{x\}}} \ell (y) = \sum_{\substack{z \in children(x)}} \bigl( |L_z| - \sum_{\substack{y \in V_z}} \ell(y) \bigr) = \sum_{\substack{z \in children(x)}} d (z)$$

Now consider a leaf $y$ of $\tau$. The outdegree of $y$ is one, and the indegree is equal to $\ell(x) + d(x) = \ell(x) + 1 - \ell(x) = 1$.
\end{proof}

\subsection{Suffix tour graph of a \$-suffix tree}
The following proposition follows from the definition of a $\$$-suffix tree.

\begin{proposition}[\cite{InferringStrings}] \label{lm:preconditions}
 If $\tau$ is a $\$$-suffix tree with a set of suffix links $\suflink$ and first letters on edges defined by a labelling function $\fl$, then
 \begin{itemize}
 	\item[(1)] For every internal explicit node $x$ in $\tau$ there exists a unique path $x = x_0 \suf x_1 \suf \ldots \suf x_k =\mbox{root}$ such that $ x_i \suf x_{i+1}$ belongs to $\suflink$ for all $i$;
 	\item[(2)] If $y$ is the end of the suffix link for $par (x)$, there is a child $z$ of $y$ such that $\fl (par (x) \rightarrow x) = \fl (y \rightarrow z)$, and the end of the suffix link for $x$ belongs to the subtree of $\tau$ rooted at $y$;
 	\item[(3)] For any node $x \in V$ the value $d(x) \ge 0$.
 \end{itemize}
\end{proposition}
If all tree conditions hold, it can be shown that

%
\begin{lemma}[\cite{InferringStrings}]\label{lm:STGisEulerian}
The tree $\tau$ is a $\$$-suffix tree iff its suffix tour graph $G$ contains a cycle $C$ which goes through the root and all leaves of $\tau$. Moreover, a string realizing $\tau$ can be inferred from $C$ in linear time. 
\end{lemma}

In more detail, the authors proved that the order of leaves in the cycle $C$ corresponds to the order of suffixes. That is, the $i^\text{th}$ leaf after the root corresponds to the $i^\text{th}$ longest suffix. Thus, the string can be reconstructed in linear time: its $i^\text{th}$ letter will be equal to the first letter on the edge in the path from the root to the $i^\text{th}$ leaf. Note that the cycle and hence the string is not necessarily unique. See \autoref{fig:stgexample} for an example.

\subsection{Suffix tour graph of a suffix tree}
\label{sec:st-stg}

We now focus on suffix tour graphs of general input trees, which are not necessarily \$-suffix trees. If the input tree $\tau$ is a suffix tree, but not a \$-suffix tree, the suffix tour graph does not necessarily contain a cycle through the root and the leaves. This is illustrated by the example in \autoref{fig:stgexample2}. We therefore have to devise a new approach.

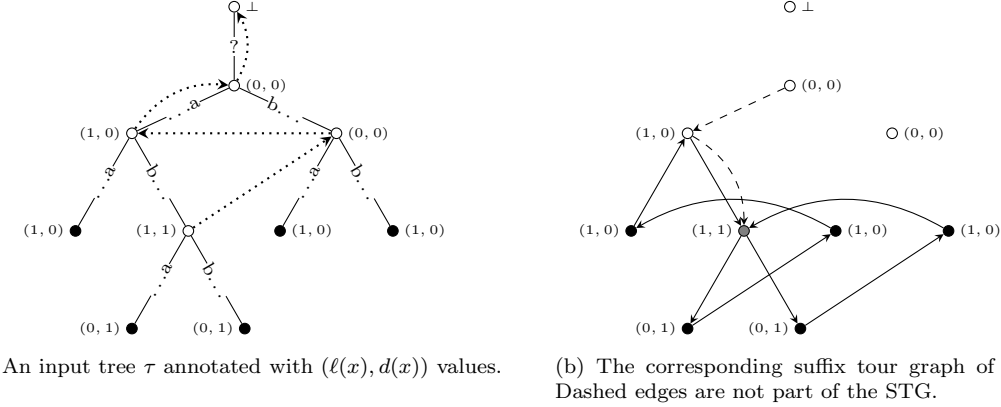
\begin{figure}
\centering
\subfloat[An input tree $\tau$ annotated with $(\ell(x),d(x))$ values.]{
\begin{tikzpicture}[scale=1.5]
\path[use as bounding box] (-2.3,1) rectangle (2.3,-2.3);
                \draw[every label/.style={draw=none,outer sep=2pt,rectangle,font=\tiny}, every node/.style={draw, circle,minimum size=4pt,inner sep=0pt}] (0,0) node[label=right:{$(0,0)$}] (root) {}
                      (root) +(90:0.7cm) node[label=right:$\perp$] (perp) {}
                      (root) +(-25:1cm) node[label=right:{$(0,0)$}] (b) {}
                        (b)+(-60:1cm) node[fill, label=right:{$(1,0)$}] (bb) {}
                        (b)+(-120:1cm) node[fill, label=right:{$(1,0)$}] (ba) {}
					  (root) +(-155:1cm) node[label=left:{$(1,0)$}] (a) {}
                        (a)+(-60:1cm) node[label=left:{$(1,1)$}] (ab) {}
                          (ab)+(-120:1cm) node[fill, label=left:{$(0,1)$}] (aba) {}
                          (ab)+(-60:1cm) node[fill, label=left:{$(0,1)$}] (abb) {}
                        (a)+(-120:1cm) node[fill, label=left:{$(1,0)$}] (aa) {}
                 ;
                \path[every node/.style={fill=white,sloped,anchor=center,allow upside down,font=\scriptsize,inner sep=0pt}]
                   (root) edge node[sloped=false, inner sep=1pt] {?} (perp)
                   (a) edge node {$\ldots$a} (root)
                   (root) edge node {b$\ldots$} (b)
				   (b) edge node {b$\ldots$} (bb)
				   (ba) edge node {$\ldots$a} (b)
				   (aa) edge node {$\ldots$a} (a)
				   (a) edge node {b$\ldots$} (ab)
				   (aba) edge node {$\ldots$a} (ab)
				   (ab) edge node {b$\ldots$} (abb)
                  ;
                  \draw[->,thick,dotted] (ab) -- (b);
                  \draw[->,thick,dotted] (b) -- (a);
                  \draw[->,thick,dotted, bend left] (a) edge (root);
                  \draw[->,thick,dotted, bend right] (root) edge (perp);
\end{tikzpicture}
}
\qquad
\subfloat[The corresponding suffix tour graph of $\tau$. Dashed edges are not part of the STG.]{
\begin{tikzpicture}[scale=1.5]
\path[use as bounding box] (-2,1) rectangle (2,-2.3);
                \draw[every label/.style={draw=none,outer sep=2pt,rectangle,font=\tiny}, every node/.style={draw, circle,minimum size=4pt,inner sep=0pt}] (0,0) node[label=right:{$(0,0)$}] (root) {}
                      (root) +(90:0.7cm) node[label=right:$\perp$] (perp) {}
                      (root) +(-25:1cm) node[label=right:{$(0,0)$}] (b) {}
                        (b)+(-60:1cm) node[fill, label=right:{$(1,0)$}] (bb) {}
                        (b)+(-120:1cm) node[fill, label=right:{$(1,0)$}] (ba) {}
					  (root) +(-155:1cm) node[label=left:{$(1,0)$}] (a) {}
                        (a)+(-60:1cm) node[fill=gray,label=left:{$(1,1)$}] (ab) {}
                          (ab)+(-120:1cm) node[fill, label=left:{$(0,1)$}] (aba) {}
                          (ab)+(-60:1cm) node[fill, label=left:{$(0,1)$}] (abb) {}
                        (a)+(-120:1cm) node[fill, label=left:{$(1,0)$}] (aa) {}
                 ;
                 \draw[->,bend right] 
                 (bb) edge (ab)
                 (ba) edge (aa)
                 ;
                 \draw[->]                 
                 (a) edge (ab)
                 (ab) edge (aba)
                 (ab) edge (abb)
                 (aba) edge (ba)
                 (abb) edge (bb)
                 (aa) edge (a)
                 ;
             	 \draw[dashed, ->] (root) edge (a);
             	 \draw[dashed, ->, bend left] (a) edge (ab); 
\end{tikzpicture}
}
\caption{This example shows the suffix tour graph (b) of an input tree $\tau$ (a), which is a suffix tree (it can be realized for example by the string \texttt{ababaa}), but not a \$-suffix tree. Contrary to the example in \autoref{fig:stgexample}, the suffix tour graph (b) does not contain a cycle going through the root and the leaves.\label{fig:stgexample2}}
\end{figure}

The high level idea of our solution is to try to augment the input tree so that the augmented tree is a \$-suffix tree. More precisely, we will try to augment the suffix tour graph of the tree to obtain a suffix tour graph of a \$-suffix tree. It will be essential to understand how the suffix tour graphs of suffix trees and \$-suffix trees are related. 

Let $\ST$ and $\ST_\$$ be the suffix tree and the \$-suffix tree of a string. We call a leaf of $\ST_\$$ a \emph{$\$$-leaf} if the edge ending at it is labeled by a single letter $\$$. Note that to obtain $\ST_\$$ from $\ST$ we must add all $\$$-leaves, their parents, and suffix links between the consecutive parents to $\ST$. 
We denote the deepest \$-leaf by $s$. 

An internal node $x$ of a suffix tour graph has $d(x)$ incoming arcs produced from edges and $\ell(x)$ incoming arcs produced from suffix links. All arcs outgoing from $x$ are produced from edges, and there are $d(x) + \ell(x)$ of them since suffix tour graphs are Eulerian graphs. A leaf $x$ of a suffix tour graph has $d(x)$ incoming arcs produced from edges, $\ell(x)$ incoming arcs produced from suffix links, and one outgoing arc produced from a suffix link. Below we describe what happens to the values $d(x)$ and $\ell(x)$, and to the outgoing arcs produced from suffix links. These two things define the changes to the suffix tour graph.

\begin{lemma}
\label{lm:l(v)}
For the deepest \$-leaf $s$ we have $\ell(s) = 0$ and $d(s) = 1$. The $\ell$-values of other \$-leaves are equal to one, and their $d$-values are equal to zero.
\end{lemma}
\begin{proof}
Suppose that $\ell(s) = 1$. Then there is a leaf $y$ such that $par_\$(y) \suf par_\$(s)$ is a suffix link in $\suflink$, and the first letter on the edge from $par_\$(y)$ to $y$ is \$. That is, $y$ is a \$-leaf and its string depth is bigger than the string depth of $s$, which is a contradiction. Hence, $\ell(s) = 0$ and therefore $d(s) = 1$. The parent of any other $\$$-leaf $y$ will have an incoming suffix link from the parent of the previous $\$$-leaf and hence $\ell(y) = 1$ and $d(y) = 0$.
\end{proof}

\noindent The important consequence of \autoref{lm:l(v)} is that in the suffix tour graph of $ST_\$$ all the \$-leaves are connected by a path starting in the deepest \$-leaf and ending in the root.

Next, we consider nodes that are explicit in $\ST$ and $\ST_\$$. If a node~$x$ is explicit in both trees, we denote its (explicit) parent in $\ST$ by $par(x)$ and in $\ST_\$$~--- by $par_\$(x)$. Below in this section we assume that each edge of $\ST$ contains at most one implicit suffix node at distance one from its parent.

\begin{lemma}\label{lm:bothexplicitorimplicit}
Consider a node $x$ of $\ST$. If a leaf $y$ contributes to $\ell(x)$ either in $ST$ or $ST_\$$, and $par_\$(y)$ and $par_\$(x)$ are either both explicit or both implicit in $ST$, then $y$ contributes to $\ell(x)$ in both trees.
\end{lemma}
\begin{proof}
If $par_\$(y)$ and $par_\$(x)$ are explicit, the claim follows straightforwardly.

Consider now the case when $par_\$(y)$ and $par_\$(x)$ are implicit. Suppose first that $y$ contributes to $\ell(x)$ in $\ST_\$$. Then the labels of $par_\$ (y)$ and $par_\$ (x)$ are $L a$ and $L[2..]a$ for some string $L$ and a letter $a$. Remember that distances between $par_\$(y)$ and $par (y)$ and between $par_\$ (x)$ and $par (x)$ are equal to one. Therefore, labels of $par(y)$ and $par(x)$  are $L$ and $L[2..]$, and the first letters on edges $par(x) \rightarrow x$ and $par(y) \rightarrow y$ are equal to $a$. Consequently, $par(y) \suf par(x)$ is a suffix link, and $y$ contributes to $\ell(x)$ in $\ST$ as well.

Now suppose that $y$ contributes to $\ell(x)$ in $\ST$. Then the labels of $par(y)$ and $par(x)$ are $L$ and $L[2..]$, and the first letters on the edges $par(y)\rightarrow y$ and $par(x)\rightarrow x$ are equal to some letter $a$. This means that the labels of $par_\$ (y)$ and $par_\$ (x)$ are $L a$ and $L[2..]a$, and hence there is a suffix link from $par_\$ (y)$ to $par_\$ (x)$. Since $y$ and $x$ are not \$-leaves, $y$ contributes to $\ell(x)$ in $\ST_\$$.
\end{proof}

Before we defined the deepest \$-leaf $s$. If the parent of $s$ is implicit in $ST$, the changes between $ST$ and $ST_\$$ are more involved. To describe them, we first need to define the \emph{twist node}. Let $p$ be the deepest explicit parent of any \$-leaf in $ST$. The node that precedes $p$ in the suffix chain is thus an implicit node in $ST$, i.e., it has two children in $ST_\$$, one which is a \$-leaf and another node $y$, which is either a leaf or an internal node. If $y$ is a leaf, let $t$ be the child of $p$ such that $y$ contributes to $\ell(t)$. We refer to $t$ as the \emph{twist} node. 

\begin{lemma}\label{lm:twistnode}
Let $x$ be a node of $\ST$. Upon transition from $\ST$ to $\ST_\$$, the $\ell$-value of $x = t$ increases by one and the $\ell$-value of its parent decreases by one. If $par_\$(x)$ is an implicit node of $\ST$, then $\ell(x)$ decreases by $\ell(par_\$(x))$. Otherwise, $\ell(x)$ does not change.
\end{lemma}
\begin{proof}
The value $\ell(x)$ can change when (1) A leaf $y$ contributes to $\ell(x)$ in $\ST_\$$, but not in $\ST$; or (2) A leaf $y$ contributes to $\ell(x)$ in $\ST$, but not in $\ST_\$$.

In the first case the nodes $par_\$(y)$ and $par_\$(x)$ cannot be both explicit or both implicit. Moreover, from the properties of suffix links we know that if $par_\$(y)$ is explicit in $\ST$, then $par_\$(x)$ is explicit as well~\cite{Gusfield}. Consequently, $par_\$(y)$ is implicit in $\ST$, and $par_\$(x)$ is explicit. Since $par_\$(x)$ is the first explicit suffix node and $y$ is a leaf that contributes to $\ell(x)$, we have $x = t$, and $\ell(x) = \ell(t)$ in $ST_\$$ is bigger than $\ell(t)$ in $\ST$ by one (see~\autoref{fig:subcases-b}).

\begin{figure}
\centering
\subfloat[]{
\begin{tikzpicture}[scale=0.9]
  \draw[every node/.style={draw, circle,minimum size=4pt,inner sep=0pt}] (2,4) node (root) {}
  	(root) +(-110:2cm) node (a) {}
  		(a) +(-110:2cm) node[fill] (b) {}
  		(a) +(-140:1.5cm) node[fill] (leaf) {}
  ;
  \node[right] at (a) {\tiny{$par_\$(y)$}};
  \node[left] at (b) {\tiny{$y$}};
  \path[every node/.style={fill=white,sloped,anchor=center,allow upside down,font=\scriptsize,inner sep=0pt}]
    (root) edge node {a} (a)
    (a) edge node {b\ldots} (b)
    (a) edge node {$\$$} (leaf)
  ;
%
  \draw[every node/.style={draw, circle,minimum size=4pt,inner sep=0pt}] (4,4) node (root1) {}
  	(root1) +(-110:2cm) node (a1) {}
  		(a1) +(-110:2cm) node (b1) {}
  		(a1) +(-140:1.5cm) node[fill] (leaf1) {}
  		(a1)+(-80:1.2cm) node (c1) {}
  ;
  \node[right] at (a1) {\tiny{$par_\$(t)$}};
  \node[right] at (b1) {\tiny{$t$}};
  \path[every node/.style={fill=white,sloped,anchor=center,allow upside down,font=\scriptsize,inner sep=0pt}]
    (root1) edge node {a} (a1)
    (a1) edge node {b\ldots} (b1)
    (a1) edge node {$\$$} (leaf1)
    (a1) edge node {c\ldots} (c1)
  ; 
  %
	\draw[->,dotted, thick, black, shorten <= 0.2pt, shorten >= 0.2pt] (root) to[out=60,in=120] (root1);
	\draw[->,dotted,  thick, black, shorten <= 0.2pt, shorten >= 0.2pt] (a) to[out=60,in=120] (a1);
%
	\draw[->, very thick, gray, shorten <= 0.2pt, shorten >= 0.2pt] (b)--(b1);
%
  \draw[every node/.style={draw, circle,minimum size=4pt,inner sep=0pt}] (-2.1,4) node (root) {}
  	(root) +(-110:4cm) node[fill] (b) {}
  ;
  \node[right] at (root) {\tiny{$par(y)$}};
  \node[left] at (b) {\tiny{$y$}};
  \path[every node/.style={fill=white,sloped,anchor=center,allow upside down,font=\scriptsize,inner sep=0pt}]
    (root) edge node {a\ldots} (b)
  ;
%
  \draw[every node/.style={draw, circle,minimum size=4pt,inner sep=0pt}] (-0.1,4) node (root1) {}
  	(root1) +(-110:2cm) node (a1) {}
  		(a1) +(-110:2cm) node (b1) {}
  		(a1)+(-80:1.2cm) node (c1) {}
  ;
  \node[right] at (a1) {\tiny{$par(t)$}};
  \node[right] at (b1) {\tiny{$t$}};
  \path[every node/.style={fill=white,sloped,anchor=center,allow upside down,font=\scriptsize,inner sep=0pt}]
    (root1) edge node {a} (a1)
    (a1) edge node {b\ldots} (b1)
    (a1) edge node {c\ldots} (c1)
  ; 
  %
	\draw[->,dotted, thick, black, shorten <= 0.2pt, shorten >= 0.2pt] (root) to[out=60,in=120] (root1);
	%
	\draw[->,very thick, gray, shorten <= 0.2pt, shorten >= 0.2pt] (b) to[out=60,in=170] (a1);
	\draw[->,very thick, gray, shorten <= 0.2pt, shorten >= 0.2pt] (a1) to[out=200,in=100] (b1);	
\end{tikzpicture}
  \label{fig:subcases-b}	
}
\subfloat[]{
\begin{tikzpicture}[scale=0.9]
  \draw[every node/.style={draw, circle,minimum size=4pt,inner sep=0pt}] (2,4) node (root) {}
  	(root) +(-110:4cm) node[fill] (a) {}
  ;
  
  \node[right] at (root) {\tiny{$par_\$(y)$}};
  \node[left] at (a) {\tiny{$y$}};
  
  \path[every node/.style={fill=white,sloped,anchor=center,allow upside down,font=\scriptsize,inner sep=0pt}]
    (root) edge node {a\ldots} (a)
  ;

  \draw[every node/.style={draw, circle,minimum size=4pt,inner sep=0pt}] (4,4) node (root1) {}
  	(root1) +(-110:2cm) node (a1) {}
  		(a1) +(-110:2cm) node (b1) {}
  		(a1) +(-140:1.5cm) node[fill] (leaf1) {}
  ;
  
  \node[right] at (a1) {\tiny{$par_\$(x)$}};
  \node[right] at (b1) {\tiny{$x$}};
  
  \path[every node/.style={fill=white,sloped,anchor=center,allow upside down,font=\scriptsize,inner sep=0pt}]
    (root1) edge node {a} (a1)
    (a1) edge node {b\ldots} (b1)
    (a1) edge node {\$} (leaf1)
  ; 
  
	\draw[->,dotted, thick, black, shorten <= 0.2pt, shorten >= 0.2pt] (root) to[out=60,in=120] (root1);
	
	\draw[->,very thick, gray, shorten <= 0.2pt, shorten >= 0.2pt] (a) to[out=60,in=170] (a1);
	\draw[->,very thick, gray, shorten <= 0.2pt, shorten >= 0.2pt] (a1) to[out=300,in=40] (b1);
	
  \draw[every node/.style={draw, circle,minimum size=4pt,inner sep=0pt}] (-1.5,4) node (root) {}
  	(root) +(-110:4cm) node[fill] (a) {}
  ;
  
  \node[right] at (root) {\tiny{$par(y)$}};
  \node[left] at (a) {\tiny{$y$}};
  
  \path[every node/.style={fill=white,sloped,anchor=center,allow upside down,font=\scriptsize,inner sep=0pt}]
    (root) edge node {a\ldots} (a)
  ;

  \draw[every node/.style={draw, circle,minimum size=4pt,inner sep=0pt}] (0.5,4) node (root1) {}
  	(root1) +(-110:4cm) node (b1) {}
  ;
  
  \node[right] at (root1) {\tiny{$par(x)$}};
  \node[right] at (b1) {\tiny{$x$}};
  
  \path[every node/.style={fill=white,sloped,anchor=center,allow upside down,font=\scriptsize,inner sep=0pt}]
    (root1) edge node {a\ldots} (b1)
  ; 
  
	\draw[->,dotted, thick, black, shorten <= 0.2pt, shorten >= 0.2pt] (root) to[out=60,in=120] (root1);
	
	\draw[->,very thick, gray, shorten <= 0.2pt, shorten >= 0.2pt] (a)--(b1);		
\end{tikzpicture}
	\label{fig:subcases-c}
}
\caption{Both figures show $\ST$ on the left and $ST_\$$ on the right.  Edges of the suffix tour graphs that change because of the twist node $t$ (\autoref{fig:subcases-b}) and because of an implicit parent (\autoref{fig:subcases-c}) are shown in grey.}
\label{fig:subcases}
\end{figure}
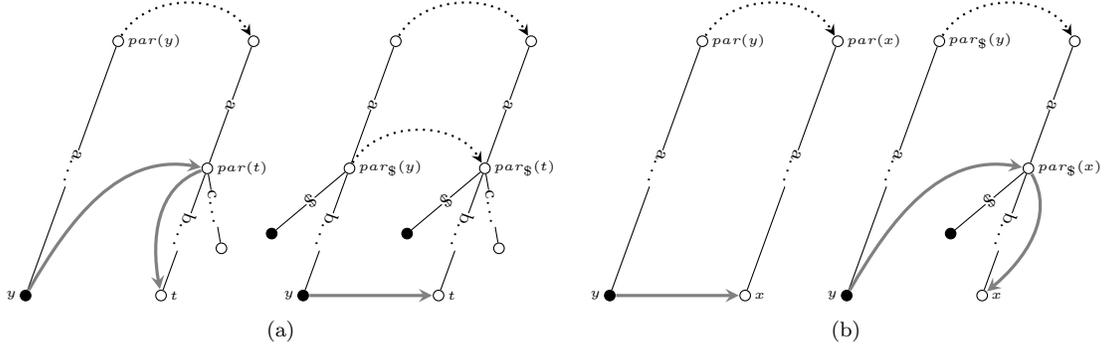

Consider one of the leaves $y$ satisfying $(2)$. In this case $par(y) \suf par(x)$ is a suffix link, and the first letters on the edges $par(y) \rightarrow y$ and $par(x) \rightarrow x$ are equal. Since $y$ does not contribute to $\ell(x)$ in $\ST_\$$, exactly one of the nodes $par_\$(y)$ and $par_\$(x)$ must be implicit in $\ST$. Hence, we have two subcases: (2a) $par_\$(y)$ is implicit in $\ST$, and $par_\$(x)$ is explicit; (2b) $par_\$(y)$ is explicit in $\ST$, and $par_\$(x)$ is implicit.

In the subcase $(2a)$ the distance between $par (y)$ and $par_\$(y)$ is one. The end of the suffix link for $par_\$(y)$ must belong to the subtree rooted at $x$. From the other hand, the string distance from $par(x)$ to the end of the suffix link is one. This means that the end of the suffix link is $x$. Consequently, $x$ is the parent of the twist node $t$, and the value $\ell(x) = \ell(par_\$(t))$ is smaller by one in $\ST_\$$ (see~\autoref{fig:subcases-b}). 

In the subcase $(2b)$ the $\ell$-value of $x$ in $\ST$ is bigger than the $\ell$-value of $x$ in $\ST_\$$ by $\ell(par_\$(x))$, as all leaves contributing to $par_\$(x)$ in $\ST_\$$, e.g. $y$, switch to $x$ in $ST$ (see~\autoref{fig:subcases-c}). 
\end{proof}

\begin{lemma}
Let $x$ be a node of $\ST$. Upon transition from $\ST$ to $\ST_\$$, the value $d(x)$ of a node $x$ such that $par_\$(x)$ is implicit in $\ST$ increases by $\ell (par_\$(x))$. If $x$ is the twist node $t$, its $d$-value decreases by one. Finally, the $d$-values of all ancestors of the deepest \$-leaf $s$ increase by one. 
\end{lemma}
\begin{proof}
Remember that $d(x) = |L_x| - \sum_{\substack{y \in V_x}} \ell (y)$. If $par_\$ (x)$ is implicit in $\ST$, $\ell(x)$ decreases by $\ell(par_\$(x))$, i.e. $d(x)$ increases by $\ell(par_\$(x))$. Note that $d$-values of ancestors of $x$ are not affected since for them the decrease of $\ell(x)$ is compensated by the presence of $par_\$(x)$. The value $\ell(t)$ increases by one and results in decrease of $d(t)$ by one, but for other ancestors of $t$ increase of $\ell(t)$ will be compensated by decrease of $\ell (par_\$ (t))$. 

The value $\ell(s) = 0$ and the $\ell$-values of other \$-leaves are equal to one. Consequently, when we add the \$-leaves to $\ST$, $d$-values of ancestors of $s$ increase by one, and $d$-values of ancestors of other \$-leaves are not affected.  
\end{proof}

\begin{lemma}
Let $par_\$(x)$ be an implicit parent of a node $x \in \ST$. Then $d(par_\$(x))$ in $ST_\$$ is equal to $d(x)$ in $\ST$ if the node $par_\$(x)$ is not an ancestor of $s$, and $d(x)+1$ otherwise.
\end{lemma}
\begin{proof}
First consider the case when $par_\$(x)$ is not an ancestor of $s$. Remember that the suffix tour graph is an Eulerian graph. The node $par_\$(x)$ has $\ell(par_\$(x))$ incoming arcs produced from suffix links and $d(x)$ outgoing arcs produced from edges. Hence it must have $d(x) - \ell(par_\$(x))$ incoming arcs produces from edges, and this is equal to $d(x)$ in $\ST$. If $par_\$(x)$ is an ancestor of $s$, the $d$-value must be increased by one as in the previous lemma.
\end{proof}

Speaking in terms of suffix tour graphs, we make local changes when the node is the twist node $t$ or when the parent of a node is implicit in $\ST$, and add a cycle from the root to $s$ (increase of $d$-values of ancestors of $s$) and back via all $\$$-leaves. 


\section{A Suffix Tree Decision Algorithm\label{sec:algorithm}}
Given a tree $\tau= (V, E)$ annotated with a set of suffix links and a labelling function, we want to decide whether there is a string $S$ such that $\tau$ is the suffix tree of $S$ and it has all the properties described in~\autoref{lm:good_implicit_ST}. 

We assume that $\tau$ satisfies~\autoref{lm:preconditions}(1) and~\autoref{lm:preconditions}(2), which can be verified in linear time. We will not violate this while augmenting $\tau$. If $\tau$ is a suffix tree, the string depth of a node equals the length of the suffix link path starting at it. Consequently, string depths of all explicit internal nodes and lengths of all internal edges can be found in linear time. 

We replace the original problem with the following one: Can $\tau$ be augmented to become a \$-suffix tree? The deepest \$-leaf $s$ can either hang from a node of $\tau$, or from an implicit suffix node $par_\$(s)$ on an edge of $\tau$. In the latter case the distance from $par_\$(s)$ to the upper end of the edge is equal to one. That is, there are $O(n)$ possible locations of~$s$. For each of the locations we consider a suffix link path starting at its parent. The suffix link paths form a tree which we refer to as the suffix link tree. The suffix link tree can be built in linear time: For explicit locations the paths  already exist, and for implicit locations we can build the paths following the suffix link path from the upper end of the edge containing a location and exploiting the knowledge about lengths of internal edges. (Of course, if we see a node encountered before, we stop.)

If $\tau$ is a suffix tree, then it is possible to augment it so that its suffix tour graph will satisfy~\autoref{lm:preconditions}(3) and~\autoref{lm:STGisEulerian}. We remind that~\autoref{lm:preconditions}(3) says that for any node $x$ of the suffix tour graph $d(x)\ge 0$, and~\autoref{lm:STGisEulerian} says that the suffix tour graph contains a cycle going through the root and all leaves. We show that each of the conditions can be verified for all possible ways to augment $\tau$ by a linear time traverse of $\tau$ or the suffix link tree. We start with~\autoref{lm:preconditions}(3). 

\begin{lemma}
If $\tau$ can be augmented to become a \$-suffix tree, then $\forall x \; d(x) \ge -1$.
\end{lemma}
\begin{proof}
The value $d(x)$ increases only when $x$ is an ancestor of $s$ or when $par_\$(x)$ is implicit in $\ST$. In the first case it increases by one. Consider the second case. Remember that $d(par_\$(x))$ is equal to $d(x)$ or to $d(x) + 1$ if it is an ancestor of~$s$. Since in a \$-suffix tree all $d$-values are non-negative, we have $d(x) \ge -1$ for any node $x$.
\end{proof}
\paragraph{Step 1.} We first compute all $d$-values and all $\ell$-values. If $d(x) \le -2$ for some node $x$ of~$\tau$, then $\tau$ cannot be augmented to become a \$-suffix tree and hence it is not a suffix tree. From now on we assume that $\tau$ does not contain such nodes. All nodes $x$ with $d(x) = -1$, except for at most one, must be ancestors of $s$. If there is a node with a negative $d$-value that is not an ancestor of $s$, then it must be the lower end of the edge containing $par_\$(s)$, and the $d$-value must become non-negative after we augment $\tau$. 

We find the deepest node $x$ with $d(x) = -1$ by a linear time traverse of $\tau$. All nodes with negative $d$-values must be its ancestors, which can be verified in linear time.
%
%
If this is not the case, $\tau$ is not a suffix tree. Otherwise, the possible locations for the parent of $s$ are descendants of $x$ and the implicit location on the edge to $x$ if $d(x)+\ell(x)$, the $d$-value of $x$ after augmentation, is at least zero. We cross out all other locations.

\paragraph{Step 2.} For each of the remaining locations we consider the suffix link path starting at its parent. If the implicit node $q$ preceding the first explicit node $p$ in the path belongs to a leaf edge then the twist node $t$ is present in $\tau$ and will be a child of $p$. We cannot tell which child though, since we do not know the first letter on the leaf edge outgoing from $q$. However, we know that $d(t)$ decreases by~$1$ after augmentation, and hence $d(t)$ must be at least $0$. Moreover, if $d(t)=0$ the twist node $t$ must be an ancestor of~$s$ to compensate for the decrease of~$d(t)$. 

In other words, a possible location of $s$ is crossed out if the twist node $t$ is present but $p$ has no child $t$ that satisfies $d(t) > 0$ or $d(t) = 0$ and $t$ is ancestor of~$s$. For each of the locations of $s$ we check if $t$ exists, and if it does, we find~$p$~(i1). This can be done in linear time in total by a traverse of the suffix link tree. We also compute for every node if it has a child $u$ such that $d(u) > 0$ (i2). Finally, we traverse $\tau$ in the depth-first order while testing the current location of $s$. During the traverse we remember, for any node on the path to $s$, its child which is an ancestor of $s$ (i3). With the information (i1), (i2), and (i3), we can determine if we cross out a location of $s$ in constant time, and hence the whole computation takes linear time.

\paragraph{Step 3.} We assume that the suffix tour graph of $\tau$ is an Eulerian graph, otherwise $\tau$ is not a suffix tree by~\autoref{lm:STGisEulerian}. This condition can be verified in linear time. When we augment $\tau$, we add a cycle $C$ from the root to the deepest \$-leaf $s$ and back via \$-leaves. The resulting graph will be an Eulerian graph as well, and one of its connected components (cycles) must contain the root and all leaves of $\tau$. 

We divide $C$ into three segments: the path from the root to the parent $par(x)$ of the deepest node $x$ with $d(x) = -1$, the path from $par(x)$ to $s$, and the path from $s$ to the root. We start by adding the first segment to the suffix tour graph. This segment is present in the cycle $C$ for any choice of $s$, and it might actually increase the number of connected components in the graph. (Remember that if $C$ contains an edge $x \rightarrow y$ and the graph contains an edge $y \rightarrow x$, then the edges eliminate each other.)

The second segment cannot eliminate any edges of the graph, and if it touches a connected component then all its nodes are added to the component containing the root of $\tau$. Since the third segment contains the \$-leaves only, the second segment must go through all connected components that contain leaves of $\tau$. We paint nodes of each of the components into some color. And then we perform a depth-first traverse of $\tau$ maintaining a counter for each color and the total number of distinct colors on the path from the root to the current node. When a color counter becomes equal to zero, we decrease the total number of colors by one, and when a color counter becomes positive, we increase the total number of colors by one. If a possible location of $s$ has ancestors of all colors, we keep it.

\begin{lemma}
\label{lm:ST}
The tree $\tau$ is a suffix tree iff there is a survived location of $s$.
\end{lemma}
\begin{proof}
If there is such a location, then for any $x$ in the suffix tour graph of the augmented tree we have $d(x) \ge 0$ and there is a cycle containing the root and all leaves. We are still to apply the local changes caused by implicit parents. Namely, for each node $x$ with an implicit parent the edge from $y$ to $x$ is to be replaced by the path $y, par_\$(x), x$ (see~\autoref{fig:subcases-c}). The cycle can be re-routed to go via the new paths instead of the edges, and it will contain the root and the leaves of $\tau$. Hence, the augmented tree is a \$-suffix tree and $\tau$ is a suffix tree. 

If $\tau$ is a suffix tree, then it can be augmented to become a \$-suffix tree. The parent of $s$ will survive the selection process.
\end{proof}

Suppose that there is such a location. Then we can find the parent of the twist node if it exists. The parent must have a child $t$ such that either $d(t) > 0$ or $d(t) = 0$ and $t$ is an ancestor of $s$, and we choose $t$ as the twist node. Let the first letter on the edge to the twist node be $a$. Then we put the first letter on all new leaf edges caused by the implicit nodes equal to $a$. The resulting graph will be the suffix tour graph of a \$-suffix tree. We can use the solution of I~et~al.~\cite{InferringStrings} to reconstruct a string $S\$$ realizing this \$-suffix tree in linear time. The tree $\tau$ will be a suffix tree of the string $S$. This completes the proof of~\autoref{thm:main}.

\section{Conclusion and Open Problems}
We have proved several new properties of suffix trees, including an upper bound of $n-1$ on the length of a shortest string $S$ realizing a suffix tree $\tau$ with $n$ nodes. As noted this bound is tight in terms of $n$, since the number of leaves in $\tau$, which can be $n-1$, provides a trivial lower bound on the length of $S$.

Using these properties, we have shown how to decide if a tree $\tau$ with $n$ nodes is a suffix tree in $O(n)$ time, provided that the suffix links of internal nodes and the first letter on each edge is specified. It remains an interesting open question whether the problem can be solved without first letters or, even, without suffix links (i.e., given only the tree structure).

Our results imply that the set of all \$-suffix trees is a proper subset of the set all of suffix trees (e.g., the suffix tree of a string $abaabab$ is not a \$-suffix tree by~\autoref{lm:STGisEulerian}), which in turn is a proper subset of the set of all trees (consider, e.g., \autoref{fig:not_st} or simply a path of length $2$).


\bibliographystyle{abbrv}
\bibliography{main}

\end{document}